\documentclass[10pt,journal]{IEEEtran}
\usepackage{moreverb}
\usepackage{color,soul}
\usepackage{latexsym}
\usepackage{graphicx,subfigure}
\usepackage{times}
\usepackage{amssymb}
\usepackage{amsmath}
\usepackage{algorithmic}
\usepackage{algorithm}
\usepackage{amsthm}
\usepackage{amssymb}

\usepackage{physics}

\usepackage[T1]{fontenc}
\usepackage[latin1]{inputenc}
\usepackage[table]{xcolor}
\usepackage[bookmarksopen,bookmarksnumbered]{hyperref}

\newtheorem{theorem}{Theorem}

\usepackage{latexsym}
\usepackage{graphicx}
\usepackage{times}
\usepackage{amssymb}
\usepackage{latexsym}
\usepackage{amsmath}
\usepackage{algorithmic}
\usepackage{algorithm}
\usepackage{amsthm}

\usepackage{color, colortbl}

\usepackage[table]{xcolor}

\ifCLASSOPTIONcompsoc
\usepackage[nocompress]{cite}
\else
\usepackage{cite}
\fi

\begin{document}
\title{Quantum Secure Protocols for Multiparty Computations}
\author{Tapaswini Mohanty, Vikas Srivastava, Sumit Kumar Debnath, Pantelimon St\u anic\u a
\thanks{Tapaswini Mohanty is with the Department of Mathematics, National Institute of Technology, Jamshedpur 831 014, India (e-mail: mtapaswini37@gmail.com). }

\thanks{Vikas Srivastava is with the Department of Mathematics, National Institute of Technology, Jamshedpur 831 014, India (e-mail: vikas.math123@gmail.com). }

\thanks{Sumit Kumar Debnath is with the Department of Mathematics, National Institute of Technology, Jamshedpur 831 014, India (e-mail: sd.iitkgp@gmail.com, sdebnath.math@nitjsr.ac.in). }

\thanks{Pantelimon St\u anic\u a is with the Department of Applied Mathematics, Naval Postgraduate School, Monterey, CA 93943, USA; (email: pstanica@nps.edu).} 
}

\markboth{}%
{}

\maketitle
\begin{abstract}
	Secure multiparty computation (MPC) schemes allow two or more parties to conjointly compute a function on their private input sets while revealing nothing but the output. Existing state-of-the-art number-theoretic-based designs face the threat of attacks through quantum algorithms. In this context, we present secure MPC protocols that can withstand quantum attacks. We first present the design and analysis of an information-theoretic secure oblivious linear evaluation (OLE), namely {\sf qOLE} in the quantum domain, and  show that our {\sf qOLE} is safe from external attacks. In addition, our scheme satisfies all the security requirements of a secure OLE. We further utilize {\sf qOLE} as a building block to construct a quantum-safe multiparty private set intersection (MPSI) protocol. 
\end{abstract}


\section{Introduction}
Secure multiparty computation (MPC) is crucial in safeguarding sensitive data. It allows two or more parties to jointly do the calculations on their private data without revealing anything but the output. Thus, MPC guarantees security features like privacy and confidentiality. The oblivious evaluation of a function is one of the most important primitives in cryptographic designs. In the work of Rabin~\cite{michael1981exchange}, the idea of oblivious transfer (OT) was introduced. OT considers the setting where there are two parties: sender and receiver. The sender has two bits $s_0$ and $s_1$, and the receiver can only learn one of the bits $s_b$ depending on his choice of bit $b$. Later, it was shown in~\cite{kilian1988founding} that OT can be used for oblivious evaluation of any cryptographic function.
Over the past three decades, considerable advancements have been made in the design of generic OT-based MPC protocols. However, it is worth noting that specific types of functions can be evaluated more efficiently using direct constructions, bypassing the need for MPC. Considering this view, Naor et al.~\cite{naor2006oblivious} designed the oblivious polynomial evaluation (OPE). It is a useful primitive which solves the problem of obliviously evaluating a polynomial $P$ on an input $\alpha$. To be more precise, OPE is a two-party protocol between two distrustful parties, where one party (say Bob) holds a private polynomial $P(x)$, and another party (say Alice) has a private input $\alpha$. The goal of a secure OPE protocol is that Alice obtains $P(\alpha)$ and nothing else while Bob learns nothing. Oblivious linear evaluation (OLE) (which is just a special case of OPE) has been studied as an important primitive in secure MPC schemes for garbled arithmetic circuits~\cite{damgaard2012multiparty,applebaum2014garble}. Instead of obliviously evaluating a polynomial $P$, we restrict ourselves to a linear function $f(x) = ax + b$ in OLE. As noted in~\cite{cianciullo2018efficient}, OLE can also be used as an important building block in the design of Private Set Intersection (PSI) and its variants. We aim to work on this line of research in the quantum domain. PSI is an efficient and secure MPC protocol that facilitates clients to compute the intersection of private sets, ensuring that the confidential information of one party is not known to the others. Multiparty PSI (MPSI) is a naturally generalized version of the PSI protocol employed to find common elements in multiple sets without exposing and leaking information about the data except for the intersection. 

There exist many OLE and MPSI protocols~\cite{efraim2021psimple,mahdavi2020practical,inbar2018efficient,abadi2020feather,badrinarayanan2021multi,bay2021multi,zhang2019efficient,kolesnikov2017practical} but, as known, their security relies on some number theoretic hardness assumptions~\cite{rivest1978method,kravitz1993digital,koblitz1987elliptic} such as integer factorization problem and discrete logarithm problem. Because of Shor's algorithm~\cite{shor1999polynomial}, these schemes face a huge security threat, and once big enough quantum computers are available, these existing state-of-the-art designs will become obsolete. It has led cryptographers worldwide to find alternative ways to construct MPC protocols such as OLE and MPSI. One line of research involves developing protocols based on mathematical problems resistant to quantum attacks, forming what we call post-quantum cryptography (PQC)~\cite{srivastava2023overview,micciancio2009lattice,ding2017current}. However, PQC does not provide long-term security. Quantum cryptography (QC) offers a logical solution to nullify the future threat. It is because QC provides long-term protection and safety from the threat of quantum attacks. Therefore, it is need of hour to make use of QC for the construction of secure MPC protocols.

\noindent\textbf{Related works:} Several quantum two-party private set intersection (PSI) protocols have been developed over the last decade. In 2016, Shi et al.~\cite{shi2016efficient} introduced a quantum protocol for calculating the intersection of private sets. However, their design had a vulnerability that allowed the server to manipulate the intersection results unilaterally. To address this issue, Cheng et al.~\cite{cheng2017cryptanalysis} involved a passive third party in protocol~\cite{shi2016efficient} to ensure fairness. Shi et al.~\cite{shi2015quantum} later designed a protocol for the oblivious set member decision problem, which could be applied to private set intersection and union. In 2018, Maitra et al.~\cite{maitra2018quantum} proposed a two-party computation for set intersections involving rational players. Recently, Debnath et al.~\cite{debnath2021feasible,debnath2022quantum} developed quantum two-party PSI protocols that utilized single-photon quantum resources, enhancing feasibility. Liu et al.~\cite{liu2021novel} presented a novel QPSI protocol based on the quantum Fourier transform, and an improved version employing the Hadamard gate was later introduced by Liu et al.~\cite{liu2022improved}. Furthermore, Liu et al.~\cite{liu2021quantum} developed a quantum private set intersection cardinality (PSI-CA) protocol based on a bloom filter, and Wang et al.~\cite{wang2021quantum} proposed a quantum protocol for PSI-CA and union cardinality using entanglement swapping. Shi et al.~\cite{shi2022quantum} designed a quantum PSI-CA protocol with privacy-preserving condition queries. In 2020, Liu et al.~\cite{liu2020quantum} introduced a quantum secure MPSI-CA protocol utilizing quantum transformation, measurements, and parallelism. Shi et al.~\cite{shi2020quantum} developed a novel quantum protocol for MPSI-CA and provided corresponding quantum circuits, highlighting the advantages of quantum computing parallelism and measurement randomness. Mohanty et al.~\cite{mohanty2023quantum} designed a special variant of PSI (called threshold PSI) which outputs an intersection if and only if the intersection size is greater than a given threshold value. In 2023, Imran~\cite{imran2023secure} proposed a quantum MPSI by providing a framework for transforming the PSI problem into a problem of computing the greatest common divisor (GCD). Additionally, in 2023 Shi et al.~\cite{shi2023edge} designed a scheme for  XOR computation of multiple private bits. In the following, they used it to create a secure multiparty logical AND, which in turn was used to design a quantum secure MPSI.

\noindent\textbf{Our contribution:}
The major contributions of this paper are the following:
\begin{itemize}
	\item We first present an efficient and information theoretic secure oblivious linear evaluation (namely {\sf qOLE}). To the best of our knowledge, there is only one design of OLE (Santos et al.~\cite{santos2022quantum}) in the current-state-of-the-art of the quantum domain. We show that our protocol is more efficient and practical compared to~\cite{santos2022quantum}. {\sf qOLE} involves three entities: Alice, Bob, and a third party TP. It is made up of three phases, namely, key generation phase ( {\sf qOLE.Kg}), initialization phase ({\sf qOLE.Int}), and computation phase ({\sf qOLE.Comp}). Bob is in possession of a private linear function $f(x)=ax+b$ defined over $\mathbb{Z}_p$. On the other hand, Alice holds $\alpha \in \mathbb{Z}_p$ and aims to obtain $f(\alpha)$, obliviously. In the key generation phase {\sf qOLE.Kg}, Alice, Bob, and TP shares secret keys with each other by utilizing a quantum key distribution (QKD). In the initialization phase ({\sf qOLE.Int}), TP selects a random linear function $S(x)=a_1x+b_1$ and prepares the corresponding quantum state $\ket{S(x)}$. The state $\ket{S(x)}$ is transferred to Bob in the quantum encrypted form over a quantum channel. Additionally, TP randomly selects $d$, computes $g=S(d)$, and sends the quantum encrypted forms of $\ket{d}$ and $\ket{g}$ to Alice through quantum communication. In the final computation phase ({\sf qOLE.Comp}), Alice sends $l=\alpha-d$ to Bob. In the following, Bob computes $V(x)=f(x+l)+S(x)$ and sends it to Alice. Ultimately, Alice uses $V(x)$ to obtain $f(\alpha)$. All the communication between Alice and Bob in the computation phase takes over a quantum communication channel.
	
	\item Based on {\sf qOLE}, we design a secure quantum protocol for MPSI (namely {\sf qMPSI}). {\sf qMPSI} uses {\sf qOLE} as the fundamental building block. Thus, we also show that quantum secure OLEs can be used to design MPSI protocols. In fact, the design of {\sf qMPSI} is generic, in the sense that it can be instantiated with any information theoretic secure OLE.
\end{itemize}

\section{Preliminaries}

\subsection{Quantum Encryption}
We describe in detail the procedure of the quantum one time pad ({\sf qOTP})~\cite{boykin2003optimal} among an encryptor Alice and decryptor Bob. It consists of three algorithms: KeyGen, Enc and Dec which are discussed below:

\begin{itemize}
	\item $K\leftarrow KeyGen(n)$. On input a positive integer $n$ (length of message), KeyGen outputs a $\kappa(\leq2n)$ bit key $K$ which is shared between Alice and Bob by using some quantum key distribution protocol (e.g.,~\cite{bennett2020quantum}).
	\item $E_K(|P\rangle)\leftarrow Enc(|P\rangle, K)$. On input the quantum message $|P\rangle=\otimes^n_{i=1}|p_i\rangle$ with $|p_i\rangle=a_i|0\rangle+b_i|1\rangle$ and $|a_i|^2+|b_i|^2=1$, and the key $K$, Alice encrypts the message $|P\rangle$  in the following way:

	Computes $ \otimes^n_{i=1}X^{K_{2i}}|p_i\rangle$, i.e., if the $(2i)$-th bit of $K$ is $1$ then operates the unitary	operator $X$ on the $i$-th qubit of the message $|P\rangle$; otherwise does nothing, for $i =1, \ldots, n$.

	Executes $E_K(|P\rangle)=\otimes^n_{i=1}Z^{K_{2i-1}}X^{K_{2i}}|p_i\rangle$by operating Z on the $i$-th qubit of $X^{K_{2i}}|p_i\rangle$
	if the $(2i-1)$-th bit of $K$ is $1$ for $i= 1, \ldots , n$.
	\item $|P\rangle \leftarrow Dec(E_K(|P\rangle))$. Bob decrypts $E_K(|P\rangle)$ by performing the following steps:
	
	To get $X^{K_{2i}}|p_i\rangle$, operates the unitary operator $Z$ on the $i$-th qubit of $E_K(|P\rangle)=\otimes^n_{i=1}Z^{K_{2i-1}}X^{K_{2i}}|p_i\rangle$, if the $(2i-1)$-th bit of $K$ is $1$, for $i= 1, \ldots , n$. 
	
	Computes $|P\rangle$ by operating the unitary operator $X$ on the $i$-th qubit of the $X^{K_{2i}}|p_i\rangle$, if the $(2i)$-th bit of $K$ is~$1$, else does nothing.
\end{itemize}

\section{Proposed Quantum Oblivious Linear Evaluation}
 \label{sec:qole}
We now discuss the design and analysis of quantum secure OLE (namely, {\sf qOLE}). The protocol {\sf qOLE} involves three parties: Alice, Bob, and a third party TP,  and consists of three phases: key generation phase ( {\sf qOLE.Kg}), initialization phase ({\sf qOLE.Int}), and computation phase ({\sf qOLE.Comp}). Bob holds a private linear function denoted by $f(x)=ax+b$ defined over $\mathbb{Z}_p$. Alice has a private input $\alpha \in \mathbb{Z}_p$ and aims to obtain $f(\alpha)$, obliviously. In the key generation phase {\sf qOLE.Kg}, Alice, Bob, and TP share secret keys with each other by utilizing a quantum key distribution (QKD). In the next phase ({\sf qOLE.Int}), TP selects a random linear function $S(x)$ and sends it to Bob through the quantum channel. Additionally, TP randomly selects $d$, computes $g=S(d)$, and sends $d$ and $g$ to Alice through quantum communication. In the final computation phase ({\sf qOLE.Comp}), Alice sends $l=\alpha-d$ to Bob. In the following, Bob computes $V(x)=f(x+l)+S(x)$ and sends it to Alice. Ultimately, Alice uses $V(x)$ to obtain $f(\alpha)$. All the communication between Alice and Bob in the computation phase takes over a quantum communication channel.

\begin{description}
	
	\item[] {\sf qOLE.Kg}: In the key generation phase, TP shares the keys $K_A$, $K_B$ with Alice and Bob, respectively, using a QKD (for example, see~\cite{bennett2020quantum}). Similarly, Alice and Bob also share a secret key $K_{AB}$  between themselves by utilizing a QKD.\\
	\item[]{\sf qOLE.Int}:  In this phase, the following steps are to be executed:
	\begin{enumerate}
		
		\item In the initialization phase, TP randomly selects a linear function $S(x)=a_1 x+b_1$. For ease of notation, we write this polynomial as $S(x)=(a_1,b_1)$. Let $\Bar{a}_1$, $\Bar{b}_1$ denote the bit strings corresponding to $a_1$ and $b_1$ respectively. TP converts these bit strings into qubits and prepares a quantum state $|S(x)\rangle=(|\Bar{a}_1\rangle,|\Bar{b}_1\rangle)$ corresponding to the linear function $S(x)$. Here, $|\Bar{a}_1\rangle$ and $|\Bar{b}_1\rangle$ are sequences of photons, i.e., $|\Bar{a}_1\rangle=\{|\Bar{a}_{11}\rangle,\ldots, |\Bar{a}_{1 \log_2{p}}\rangle\}$, $|\Bar{b}_1\rangle=\{|\Bar{b}_{11}\rangle,\ldots, |\Bar{b}_{1 \log_2{p}}\rangle\}$ and $\Bar{a}_{1j},\Bar{b}_{1j}\in \{0,1\}$. TP encrypts $|S(x)\rangle$ using quantum OTP with key $K_B$ and gets $|S'(x)\rangle$. 
		
		\item In the following, TP inserts some decoy particles (randomly chosen from $\{|0\rangle, |1\rangle, |+\rangle, |-\rangle\}$) into $|S'(x)\rangle$ and obtains $|S''(x)\rangle$. TP notes the position and initial states of the decoy photons and sends $|S''(x)\rangle$ to Bob.

		\item TP randomly selects $d$ and computes $S(d)=g$ (say). TP first converts $d$ and $g$ into binary string. Later, it prepares the quantum states $|d\rangle$ and $|g\rangle$ corresponding to $d$ and $g$ respectively. In the following, it encrypts $|d\rangle$ and $|g\rangle$ using quantum OTP with key $K_A$ and gets $|d'\rangle, |g'\rangle$.  TP also inserts some decoy photons (randomly chosen from $\{|0\rangle, |1\rangle, |+\rangle, |-\rangle\}$) into $|d'\rangle, |g'\rangle$ and gets $|d''\rangle, |g''\rangle$. TP notes the position and initial states of the decoy photons and sends  $|d"\rangle, |g"\rangle$ to Alice.\\
		
	\end{enumerate}
	\item[]  {\sf qOLE.Comp}: The following steps are to be performed between Alice and Bob in the computation phase:
	
	\begin{enumerate}
		\item On receiving $|d''\rangle, |g''\rangle$, Alice first checks for eavesdropping. TP provides the position and initial states of the decoy particles. Alice aborts the protocol if the error rate exceeds the predefined threshold value; otherwise, proceeds to the next step. 
		\item Alice discards all decoy photons and gets $|d'\rangle, |g'\rangle$. Then Alice decrypts $|d'\rangle, |g'\rangle$ using $K_A$ and gets $|d\rangle, |g\rangle$. Alice measures the quantum state and obtains $d, g$. 
		\item Alice computes $l=\alpha-d$ and converts it into binary string and then into a sequence of photons $|l\rangle$. She encrypts $| l \rangle$ using a quantum OTP with $K_{AB}$. To avoid eavesdropping, she adds some decoy particles randomly chosen from $\{|0\rangle, |1\rangle, |+\rangle, |-\rangle\}$ into $|l'\rangle$, and obtains $|l''\rangle$. Alice notes the position and initial states of the decoy particles, and sends $|l''\rangle$ to Bob.
		
		\item On receiving $|S''(x)\rangle$ from TP, Bob and TP checks for eavesdropping; if the error rate exceeds the predefined threshold value then they abort the protocol; otherwise, they proceed to the next step. Bob also checks for eavesdropping with Alice. If the error rate exceeds a predefined threshold then they abort the protocol, otherwise they proceed to the next step. Bob discards all decoy photons from $|S''(x)\rangle$, $|l''\rangle$ and gets  $|S'(x)\rangle$, $|l'\rangle$. Bob decrypts  $|S'(x)\rangle$, $|l'\rangle$ using keys $K_B$, $K_{AB}$, respectively, and gets $|S(x)\rangle$, $|l\rangle$. He measures the state $|S(x)\rangle$, $|l\rangle$ using a  computational basis and obtains  $S(x)$ and $l$. 
		\item Bob computes $V(x)=f(l+x)+S(x)$. He converts $V(x)$ sequence of photons  $|V(x)\rangle$. In the following, $|V(x)\rangle$ is encrypted using a quantum OTP with key $K_{AB}$ to produce $|V'(x)\rangle$. Bob inserts some decoy photons from $\{|0\rangle, |1\rangle, |+\rangle, |-\rangle\}$ into $|V'\rangle$ and prepares $|V''\rangle$. Finally, $|V''\rangle$ is sent to Alice.
		\item On receiving $|V''\rangle$ from Bob,  Alice checks for eavesdropping. If the error rate is more than the the threshold, she aborts the protocol; otherwise, she proceeds. After discarding all the decoy photons from $|V''(x)\rangle$, she obtains $|V'(x)\rangle$. Alice decrypts  $|V'(x)\rangle$ using the key $K_{AB}$ to get $|V(x)\rangle$. Now $|V(x)\rangle$ is measured using the computational basis to obtain  $V(x)$. Next,  Alice computes $V(d)-g=f(l+d)+S(d)-g=f(\alpha)$.
		
	\end{enumerate}

\end{description}

\noindent \textbf{Correctness:} It is very easy to see that {\sf qOLE} is correct, i.e., at the end of the protocol Alice correctly obtains $f(\alpha)$. This follows because
\begin{align*}
V(d)-g &=f(l+d)+S(d)-g\\
&=f(\alpha-d+d)+S(d)-g\\
&=f(\alpha)+g-g\\
&=f(\alpha).
\end{align*}

\subsection{Toy Example}
Let $p=8$ and suppose Bob uses $f(x)=2x+3$ and Alice picks $\alpha=4\in \mathbb{Z}_8$. Assume that TP shares $K_A=011101010011$ with Alice and $K_B=101001110001$ with Bob. Let the secret key of Alice and Bob be $K_{AB}=111001000101$. TP chooses $S(x)=3x+1$ randomly. He picks $d=2$ and computes $g=S(d)=7$.

\noindent Note that, $a_1=3$, $b_1=1$ so $\Bar{a}_1=011$, $\Bar{b}_1=001$, therefore, $|S(x)\rangle=|0\rangle|1\rangle|1\rangle,|0\rangle|0\rangle|1\rangle$.
Now, {\small $|S'(x)\rangle=E_{K_B}(|S(x)\rangle)=Z^1X^0(|0\rangle) Z^1X^0(|1\rangle) Z^0X^1(|1\rangle) Z^1X^1(|0\rangle) Z^0X^0(|0\rangle) Z^0X^1(|1\rangle)$.}
i.e., $|S'(x)\rangle=|0\rangle (-|1\rangle) |0\rangle (-|1\rangle)|0\rangle |0\rangle$. Suppose that decoy particles are added at the $2^{nd},5^{th},6^{th}$ and $9^{th}$ positions.
Let $|S''(x)\rangle=|0\rangle |+\rangle (-|1\rangle) |0\rangle |0\rangle |-\rangle (-|1\rangle)|0\rangle |0\rangle|0\rangle$.

\noindent Note that $|d\rangle,|g\rangle=|0\rangle|1\rangle|0\rangle, |1\rangle|1\rangle|1\rangle$
so {\small $|d'\rangle,|g'\rangle=E_{K_A}(|d\rangle, |g\rangle)=Z^0X^1(|0\rangle) Z^1X^1(|1\rangle) Z^0X^1(|0\rangle) Z^0X^1(|1\rangle) Z^0X^0(|1\rangle) Z^1X^1(|1\rangle)$} i.e., $|d'\rangle,|g'\rangle=|1\rangle|0\rangle|1\rangle, |0\rangle|1\rangle|0\rangle$. Let the positions of the decoy photons be $1,2,6$ and $9$. Therefore, $|d''\rangle,|g''\rangle=|1\rangle |1\rangle|1\rangle|0\rangle|1\rangle |+\rangle |0\rangle|1\rangle |+\rangle |0\rangle$. On receiving $|S''(x)\rangle$ from TP, Bob checks for eavesdropping and discards all the decoy photons. After this step, Bob gets $|S'(x)\rangle=|0\rangle (-|1\rangle) |0\rangle (-|1\rangle)|0\rangle |0\rangle$. Later, he computes\\
 $|S(x)\rangle$=$D_{K_B}(|S(x)\rangle)=$ {\small  $X^0Z^1(|0\rangle) X^0Z^1(-|1\rangle) X^1Z^0(|0\rangle) \\X^1Z^1(-|1\rangle)X^0Z^0(|0\rangle)  X^1Z^0(|0\rangle)$}, i.e., $|S(x)\rangle=|0\rangle |1\rangle |1\rangle |0\rangle |0\rangle |1\rangle$. After measurement, Bob gets $\Bar{a}_1=011$, $\Bar{b}_1=001$, i.e., $S(x)=3x+1$.

\noindent On receiving $|d''(x)\rangle,|g''\rangle$ from TP, Alice checks for eavesdropping and discards all the decoy photons. Alice gets $|d'\rangle,|g'\rangle=|1\rangle|0\rangle|1\rangle, |0\rangle|1\rangle|0\rangle$. In the next step, Alice computes
{\small $|d\rangle,|g\rangle=D_{K_A}(|d'\rangle,|g'\rangle)=X^1Z^0(|1\rangle) X^1Z^1(|0\rangle) X^1Z^0(|1\rangle), X^1Z^0(|0\rangle) X^0Z^0(|1\rangle) X^1Z^1(|0\rangle)$} i.e., {\small $|d\rangle,|g\rangle=|0\rangle|1\rangle|0\rangle, |1\rangle |1\rangle |1\rangle$}. Alice does the measurements and obtains $d=2$, and $g=7$. Now Alice computes $l=\alpha-d=2$, converts it into qubits and gets  $|l\rangle=|0\rangle|1\rangle|0\rangle$.
Alice computes $|l'\rangle=E_{K_{AB}}(|l\rangle)=Z^1X^1(|0\rangle) Z^1X^0(|1\rangle) Z^0X^1(|0\rangle)$ i.e.,$|l'\rangle=(-|1\rangle) (-|1\rangle) |1\rangle$. Let the positions of the decoy photons be $1,2$ and $4$,
then $|l''\rangle=|+\rangle |-\rangle (-|1\rangle) |0\rangle (-|1\rangle) |1\rangle$
Alice sends $|l''\rangle$ to Bob.

\noindent Bob checks for eavesdropping and discards all decoy photons and gets $|l'\rangle=(-|1\rangle) (-|1\rangle) |1\rangle$. Bob computes $|l\rangle=D_{K_{AB}}(|l'\rangle)= X^1Z^1(-|1\rangle) X^0Z^1(-|1\rangle) X^1Z^0(|1\rangle)$, i.e., $|l\rangle=|0\rangle |1\rangle |0\rangle$. He measures the quantum state and gets $l=2$. Now Bob computes $V(x)=f(l+x)+S(x)=f(2+x)+3x+1=2(2+x)+3+3x+1=5x$ convert it into qubits $|V(x)\rangle=|1\rangle|0\rangle|1\rangle|0\rangle|0\rangle|0\rangle$,
computes  {\small $|V'(x)\rangle=E_{K_{AB}}(|V(x)\rangle)=Z^1X^1(|1\rangle) Z^1X^0(|0\rangle) Z^0X^1(|1\rangle) Z^0X^0(|0\rangle) Z^0X^1(|0\rangle) Z^0X^1(|0\rangle)$}, i.e., $|V'(x)\rangle=|0\rangle|0\rangle|0\rangle |0\rangle|1\rangle|1\rangle$. Let the positions of the decoy photons be $3,4,5$, and $8$, then 
$|V''(x)\rangle=|0\rangle|0\rangle |0\rangle|1\rangle|0\rangle|0\rangle |0\rangle|+\rangle|1\rangle|1\rangle$.
Bob send $|V''(x)\rangle$ to Alice.

\noindent On receiving $|V''(x)\rangle$ from Bob, Alice checks for eavesdropping and discards all the decoy photons. Alice gets  $|V'(x)\rangle=|0\rangle|0\rangle|0\rangle |0\rangle|1\rangle|1\rangle$. Now, she computes
{\small $|V(x)\rangle=D_{K_{AB}}(|V'(x)\rangle)=X^1Z^1(|0\rangle) X^0Z^1(|0\rangle) X^1Z^0(|0\rangle) X^0Z^0(|0\rangle) X^1Z^0(|1\rangle) X^1Z^0(|1\rangle)$} i.e., $|V(x)\rangle=|1\rangle|0\rangle|1\rangle |0\rangle|0\rangle|0\rangle$
Alice measures and gets $V(x)=5x$. Finally, she computes $V(d)-g=V(2)-7=2-7=-5\equiv3 \pmod{8}$ to obtain $f(\alpha=4)=3$.

\subsection{Security Analysis}
\subsubsection{External attacks} 

In this setting, an external attacker wants to obtain some information $\alpha$, $f(x)$, or $f(\alpha)$ by interfering the communication channel. The attacker only gets $|S''(x)\rangle$, $|d''\rangle,|g''\rangle$, $|l''\rangle$, and $|V''(x)\rangle$. As the attacker does not know the actual position of the decoy particles, he therefore cannot obtain  $|S'(x)\rangle$, $|d'\rangle,|g'\rangle$, $|l'\rangle$, and $|V'(x)\rangle$.

We now discuss the situation when an attacker applies the \textbf{ entangled measurement attack} on a decoy photon, say $|\phi\rangle \in \{|0\rangle, |1\rangle,|+\rangle, |-\rangle\}$. On receiving the decoy particle $|\phi\rangle$ he prepares an ancillary qubit $|0\rangle_a$ and operates an oracle operator $\mathcal{U}_f|x\rangle|y\rangle \mapsto|x\rangle|y\oplus f(x)\rangle$.

Case I:

If $|\phi\rangle$ is $|0\rangle$ or $|1\rangle$,
\begin{align*}
\mathcal{U}_f|\phi\rangle|0\rangle_a =
\begin{cases}
|0\rangle|f(0)\rangle_a & \text{if $|\phi\rangle=|0\rangle$,}\\
|1\rangle|f(1)\rangle_a & \text{if $|\phi\rangle=|1\rangle$.}
\end{cases}
\end{align*}

Case II:

If $|\phi\rangle$ is $|+\rangle$ or $|-\rangle$,
\begin{align*}
\mathcal{U}_f|\phi\rangle|0\rangle_a= & \frac{\mathcal{U}_f|0\rangle|0\rangle_a\pm \mathcal{U}_f|1\rangle|0\rangle_a}{\sqrt{2}}\\
=& \frac{1}{\sqrt{2}}{\frac{|0\rangle\pm|1\rangle}{\sqrt{2}}\otimes \frac{|f(0)\rangle_a\pm |f(1)\rangle_a}{\sqrt{2}}}.
\end{align*}
By the above analysis, we see that if ${|\phi \rangle} \in \{|0\rangle, |1\rangle\}$, then the external attacker can guess correctly, but if  ${|\phi \rangle} \in \{|+\rangle, |-\rangle\}$, then the success probability is $\frac{1}{2}$. In addition, all these photons are non orthogonal so these states are indistinguishable. Therefore, an outsider fails to obtain any information.

\noindent{\bf Intercept and resend attack:} During these kinds of attacks, an adversary intercepts and resends the sequence of photons $|\Bar{S''}(x)\rangle$ to Bob by intercepting the stream of photons $|S''(x)\rangle$ sent by TP to Bob. TP had added the decoy photons before sending $|S''(x)\rangle$ to Bob. Now, the adversary is unaware of the actual position
and the state of decoy photons. Note that decoy photons are randomly chosen out of $\{|0\rangle, |1\rangle, |+\rangle, |-\rangle\}$. Therefore, any external interception can be detected with a probability $1-(\frac{3}{4})^{\delta}$. When $\delta \ggg 0$, the probability of detecting an eavesdropping converges to $1$.

\noindent {\bf Trojan Horse attacks}: {\sf qOLE} may be susceptible to two types of Trojan horse attacks: \textit{delayed photon attacks} and \textit{invisible photon attacks}. These attacks involve manipulating photons during transmission to compromise the security of the protocol. To counter these attacks, participants can incorporate specific quantum optical devices, such as wavelength quantum filters and photon number splitters, during the protocol execution. To address the issue of invisible photons that may arise during transmission, a wavelength quantum filter can be employed to filter out these photons, ensuring that only legitimate photons are considered.
Similarly, for delayed photons that may be present, photon number splitters can be utilized to split each legitimate photon, thereby detecting any delayed photons. This enables the participants to identify and mitigate the effects of delayed photon attacks.

\subsubsection{Internal Attack}
To consider internal attacks, we assume that {\sf qOTP} is information theoretic secure~\cite{boykin2003optimal}. In addition, the following results hold.
\begin{theorem}
	Alice cannot obtain $f(x)$.
\end{theorem}
\begin{proof} Alice cannot obtain $f(x)$ because of two reasons. Firstly, she does not know the actual position of decoy particles. Secondly, the computation of $f(x)$ requires the knowledge of $S(x)$, but she cannot obtain $S(x)$, since obtaining $S(x)$ from $|S''(x)\rangle$ implies a break in the information theoretic security of~{\sf qOTP}. 
\end{proof}
\begin{theorem}
	Bob cannot obtain $\alpha$ and $f(\alpha)$.
\end{theorem}
\begin{proof}
	Similar to the prior theorem, Bob can not obtain $\alpha$ and $f(\alpha)$ because of two reasons. Firstly, he also does not know the actual position of decoy particles. Secondly, the computation of $\alpha$ and $f(\alpha)$ requires the knowledge of $d$ and $g$. It is not possible to obtain any information about $d$ and $g$ because gaining any knowledge about $d$ and $g$ from $|d''\rangle$ and $|g''\rangle$ implies a break in the information theoretic security of~{\sf qOTP}. 
\end{proof}

\begin{theorem}
	TP cannot obtain $f(x)$, $\alpha$, and $f(\alpha)$.
\end{theorem}
\begin{proof}
	To obtain any knowledge $f(x)$, $\alpha$, and $f(\alpha)$, TP needs information about $V(x)$ and $l$. The information theoretic security of {\sf qOTP} makes it impossible for TP to obtain $V(x)$ and $l$ from $|V''(x)\rangle$ and $|l''\rangle$.\end{proof}

\subsection{Efficiency and Comparison}
In this section, we will discuss the communication and
computational overhead of our proposed design {\sf qOLE}. The linear function $f(x)=ax+b$ is defined over $\mathbb{Z}_p$. TP sends $2\log_2{p}$ qubits to Bob and sends $\log_2{p}$ qubits to Alice. In addition, Alice sends $\log_2{p}$ qubits to Bob and Bob sends $\log_2{p}$ qubits to Alice. Therefore, the total communication cost is $\mathcal{O}(\log_2{p})$. 
We now discuss the quantum computation required for the execution of {\sf qOLE}.  A total of $\mathcal{O}(\log_2{p})$ qubits are needed to be prepared by Alice, Bob and TP. In addition, Pauli operators $X$ and $Z$ are used for doing the quantum computation. Projective measurements of $\mathcal{O}(\log_2{p})$ single qubits are required during the initialization and computation phase.

To the best of our knowledge, there is only one OLE protocol in the quantum domain. Santos et al.~\cite{santos2022quantum} in 2022 developed the first OLE protocol in the quantum domain. Similar to {\sf qOLE},~\cite{santos2022quantum} does not rely on quantum oblivious transfer. Unlike {\sf qOLE}, the design presented in~\cite{santos2022quantum} uses \textit{high-dimensional quantum states} to obliviously compute the linear function $f(x)$. {\sf qOLE} in contrast used only single photons which are very easy to prepare and operate.~\cite{santos2022quantum} utilize the Heisenberg-Weyl operators, while {\sf qOLE} only needs two dimensional quantum gates like $X$ gate and $Z$ gate. To summarize, {\sf qOLE} is more practical and efficient when compared to the design presented in~\cite{santos2022quantum}. {\sf qOLE} only uses single photons quantum resources, and thus, has the potential to be implemented in the near future. A comparative summary of {\sf qOLE} with~\cite{santos2022quantum} is given in Table~\ref{tab:my_label}. 

\begin{table*}[ht]
	\centering
	\begin{tabular}{|p{2.8cm}|p{3.0cm}|p{3.0cm}|p{3.0cm}|p{2.0cm}|}\hline
		Protocols  & Quantum  Resources & Quantum Operators  & Quantum \newline Measurements & Security \\ \hline
		Santos et al.~\cite{santos2022quantum} & higher dimensional quantum states & Heisenberg-Weyl operators& Projective measurements in
		$d$-dimension Hilbert space& Quantum \newline UC-security\\ \hline
		This protocol & $2$-dimensional quantum states & $2$-dimensional $X$ and $Z$ gate & Projective measurements in
		$2$-dimension Hilbert space& Information theoretic security \\
		\hline
	\end{tabular}
	\caption{Comparative summary of {\sf qOLE}}
	\label{tab:my_label}
\end{table*}

\section{Proposed Quantum MPSI}
In this section, we present the construction and analysis of an MPSI protocol (namely, {\sf qMPSI}) in the quantum domain. Our design utilizes the OLE protocol ({\sf qOLE}) described in Section~\ref{sec:qole}. We first give a high level overview of the protocol, followed by a detailed explanation.\\
\noindent \textbf{A high level overview.} In the design of {\sf qMPSI}, there are $m$ parties $A_1, A_2\ldots A_m$, each having private sets $S_{A_1}, S_{A_2}, \ldots, S_{A_m}$, respectively. These sets are defined over $\mathbb{Z}_p$ with $|S_{A_j}|=n$.  In the preparation stage, each $A_i$ forms a polynomial corresponding to their private sets. Without  loss of generality, we set $A_1$ to be initiator of the protocol. At the end of the execution of {\sf qMPSI}, $A_2$ computes the desired intersection and announces it for everybody. We assume that all the parties involved in our protocol are semi-honest.

\begin{description}
	\item[]{\sf qMPSI.PrepStage}: There are $m$ parties $A_1, A_2, \ldots, A_m$ with private sets $S_{A_1}, S_{A_2}, \ldots, S_{A_m}$. Corresponding to their private sets $S_{A_j}$, each of the $A_j$ define a polynomial $P_{A_j}$ of degree $n$ such that $P_{A_j}(\gamma)=0$ for all $\gamma \in S_{A_j}$. Let $Z$ be the set of all zeros of polynomials $P_{A_1}, P_{A_2}, \ldots, P_{A_m}$ and $\mathcal{S}$ be the set of polynomials whose zeros are not in $Z$.
	Next, each party $A_j$, $j=1,2,\ldots, m$, randomly chooses polynomials $r_{A_j},r_j\in \mathcal{S}$ of degree at most $n$. The sets $A_j$, $j=1,2,\ldots,m$, mask their private polynomial $P_{A_j}$ with $r_{A_j}$ and gets $P'_{A_j}$ . $A_1$ chooses a polynomial $u_{A_1}$ of degree $n$, computes $P_1=P'_{A_1}+u_{A_1}$. $A_1,A_2,\ldots,A_m$, and  selects $\alpha_1,\alpha_2, \ldots, \alpha_{3n+1}$.\\
	
	\item[] {\sf qMPSI.Intersection}: The following steps are performed to compute the intersection:
	\begin{enumerate}
		\item $A_1$ and $A_2$ jointly execute the {\sf qOLE} protocol with the help of TP. The input of $A_1$ and $A_2$ to {\sf qOLE} is $(P^i_1,r^i_1)$ and $P'^i_{A_2}$, respectively. The output of the protocol is $P^i_2={P'^i_{A_2}}{r^i_1}+P^i_1$, where $P'^i_{A_2}=P'_{A_2}(\alpha_i)$, $P^i_1=P_1(\alpha_i)$ and $r^i_1=r_1(\alpha_i)$. $A_1$ and $A_2$ execute {\sf qOLE}  protocol, for $i=1,2,\ldots,3n+1$.
		
		\item For $j=3,\ldots,m$,
		similarly, $A_{j-1}$ and $A_{j}$  jointly execute the {\sf qOLE} protocol with the help of TP. The input of $A_{i-1}$ to {\sf qOLE} is $(P^i_{j-1},r^i_{j-1})$ while the input of $A_j$ to {\sf qOLE} is $P'^i_{A_j}$. The output of the protocol is $P^i_j={P'^i_{A_j}}{r^i_{j-1}}+P^i_{j-1}$, where $P'^i_{A_j}=P'_{A_j}(\alpha_i)$, $P^i_{j-1}=P_{j-1}(\alpha_i)$ and $r^i_{j-1}=r_{j-1}(\alpha_i)$. $A_{j-1}$ and $A_j$ execute this {\sf qOLE} protocol for $i=1,2,\ldots,3n+1$.
		
		\item $A_2, A_3, \ldots, A_m$ share a common polynomial $u$ with each other.
		\item $A_m$ computes $R^i=P^i_m+u^i={P'^i_{A_m}}{r^i_{m-1}}+P^i_{m-1}+u^i$, for $i=1,2,\ldots,3n+1$, where $u^i=u(\alpha_i)$ and sends the resulting values to $A_1$.
		\item $A_1$ subtract $u^i_{A_1}$ from
		$R^i$, for all $i=1,2,\ldots,3d+1$ and sends it to $A_2$. $A_2$ subtracts $u^i$ for all $i=1,2,\ldots, 3n+1$ and interpolates these $3n+1$ points and gets the desired intersection polynomial $P_{\cap}$. At the end of the protocol, $A_2$, outputs all $\gamma \in S_{A_2}$ for which $P_{\cap}(\gamma)=0$. $A_2$ announces the intersection.    
	\end{enumerate}
\end{description}
\subsection{Correctness}
The correctness of {\sf qMPSI} follows in a straightforward manner from the correctness of {\sf qOLE}. $A_2$ received $R^i-u^i_{A_1}$ with $i=1,2,\ldots,3n+1$ from $A_1$. $A_2$ subtracts $u^i$ for all $i=1,2,\ldots, 3n+1$ and interpolates these $3n+1$ points and gets the desired intersection polynomial $P_{\cap}$, where
\begin{align*}
P_{\cap}&=P_{A_m}r_{A_m}r_{m-1}+P_{A_{m-1}}r_{A_{m-1}}r_{m-2}\\
&\qquad+\cdots +P_{A_2}r_{A_2}r_{1}+P_{A_1}r_{A_1}.
\end{align*}
We know that $P_{A_j}$, for $j=1,2,\ldots,m$ is the polynomial corresponding to the set $S_{A_j}$, for $j=1,2,\ldots,m$, i.e., for $j=1,2,\ldots,m$ if $x\in S_{A_j}$, then $P_{A_j}(x)=0$ and $r_{A_j}$, $r_j$ are nonzero and belongs to $Z$, for all $j=1,2,\ldots,m$.

If $x\in S_{A_1}\cap S_{A_2}\cap \ldots \cap S_{A_m}$, then
$ x\in S_{A_j}$ for all $j=1,2,\ldots,m$, so
 $P_{A_j}(x)=0$ for all $j=1,2,\ldots,m$, which renders
$ P_{\cap}(x)=0$.

If $P_{\cap}(x)=0$, then
$  [P_{A_m}r_{A_m}r_{m-1}+P_{A_{m-1}}r_{A_{m-1}}r_{m-2}+\cdots+P_{A_2}r_{A_2}r_{1}+P_{A_1}r_{A_1}](x)=0$, so
 $P_{A_j}(x)=0$ for all $j=1,2,\ldots,m$ as for $j=1,2,\ldots,m$, $r_{A_j}(x)\neq0$ and $r_j(x)\neq0$, and therefore
$x\in S_{A_j}$ for all $j=1,2,\ldots,m$ i.e., $x\in S_{A_1}\cap S_{A_2}\cap \ldots \cap S_{A_m}$.

\subsection{Security Analysis}
\begin{theorem}
	$A_1$ cannot learn anything about the private sets of other parties.
\end{theorem}
\begin{proof}
	$A_1$ has $P_{A_1}$, $r_{A_1}$, $u_{A_1}$, $P_1$, and $r_1$. At the end, $A_m$ sends $P^i_m=P^i_{A_m}r^i_{A_m}r^i_{m-1}+P^i_{A_{m-1}}r^i_{A_{m-1}}r^i_{m-2}+\cdots+P^i_{A_2}r^i_{A_2}r^i_{1}+P^i_{A_1}r^i_{A_1}+u^i_{A_1}+u^i$, for $i=1,2\ldots,3n+1$ to $A_1$. $A_1$ subtracts $P^i_{A_1}r^i_{A_1}+u^i_{A_1}$ from $P^i_m$, for $i=1,2,\ldots,3n+1$ and gets $P^i_{A_m}r^i_{A_m}r^i_{m-1}+P^i_{A_{m-1}}r^i_{A_{m-1}}r^i_{m-2}+\cdots+P^i_{A_2}r^i_{A_2}r^i_{1}+u^i$, for $i=1,2,\ldots,3n+1$. Since $A_1$ does not know $u^i$ for $i=1,2,\ldots,3n+1$, therefore $A_1$ cannot get any information about the private set of other parties and also cannot obtain the intersection polynomial for computing the intersection $S_{A_2}\cap S_{A_3}\cap \ldots \cap S_{A_m}$.
	
\end{proof}
\begin{theorem}
	{\sf qMPSI} protocol is collusion resistant provided $A_1$ and $TP$ are non-collusive parties.
	
\end{theorem}
\begin{proof}
	Suppose $A_2, A_3,\ldots,A_{k-1},A_{k+1},\ldots,A_m$ want to know some information about the private set $S_{A_k}$ of $A_k$. If they can extract $P^i_{A_k}$ for all $i=1,2,\ldots,3n+1$, they can get $P_{A_k}$ and from $P_{A_k}$ they may obtain $S_{A_k}$. From $P^i_m=P^i_{A_m}r^i_{A_m}r^i_{m-1}+P^i_{A_{m-1}}r^i_{A_{m-1}}r^i_{m-2}+\cdots+P^i_{A_2}r^i_{A_2}r^i_{1}+P^i_{A_1}r^i_{A_1}+u^i_{A_1}+u^i$ they subtract their inputs and gets $P^i_{A_k}r^i_{A_k}r^i_{k-1}+P^i_{A_1}r^i_{A_1}+u^i_{A_1}+u^i$ for all $i=1,2,\ldots,3n+1$. As they only know $r^i_{k-1}$ for all $i=1,2,\ldots,3n+1$, therefore cannot obtain $P^i_{A_k}$ for any $i=1,2,\ldots,3n+1$. Hence, they cannot obtain any information about the private set of $A_k$.
\end{proof}
\begin{theorem}
	TP cannot get any information about the private set of any party.
\end{theorem}
\begin{proof}
	TP only initializes the {\sf qOLE} protocol. We already showed that any external party cannot obtain anything about anyone's private set and TP cannot obtain any information during the {\sf qOLE} protocol. When $A_m$ sends $P^i_m=P^i_{A_m}r^i_{A_m}r^i_{m-1}+P^i_{A_{m-1}}r^i_{A_{m-1}}r^i_{m-2}+\cdots+P^i_{A_2}r^i_{A_2}r^i_{1}+P^i_{A_1}r^i_{A_1}+u^i_{A_1}+u^i$ to $A_1$ and $A_1$ sends $P^i_{A_m}r^i_{A_m}r^i_{m-1}+P^i_{A_{m-1}}r^i_{A_{m-1}}r^i_{m-2}+\cdots+P^i_{A_2}r^i_{A_2}r^i_{1}+P^i_{A_1}r^i_{A_1}+u^i$ to $A_2$, TP may intercept to gain the intersection of the private sets of $A_1,A_2,\ldots,A_m$. Since TP does not know $u^i_{A_1}$ and $u^i$ for any $i=1,2,\ldots.3n+1$, therefore it obtains nothing about the private sets or intersection of the private sets.
\end{proof}

\subsection{Efficiency Analysis and Comparison}
In this section, we will discuss the communication and computational overhead of our proposed design {\sf qMPSI}. The private polynomials are defined over $\mathbb{Z}_p$. {\sf qOLE} protocol is used as a building block in the deisgn of {\sf qMPSI}. {\sf qOLE} is executed $(m-1)(3n+1)$ number of times during the execution of {\sf qMPSI}, where $m$ is the number of parties and $n$ is the size of the private set. Therefore, the communication and computation costs of {\sf qMPSI} are $\mathcal{O}(mn\log_2{p})$ and $\mathcal{O}(mn\log_2{p})$, respectively.

To the best of our knowledge, there are two MPSI protocol in the quantum domain. Shi et al.~\cite{shi2023edge} in 2023 developed a quantum MPSI which ensures perfect security. Shi et al.~\cite{shi2023edge} introduced a semi-honest edge server and two non-collusive fog nodes, and design a secure and feasible edge-assisted quantum protocol for MPSI. The design of~\cite{shi2023edge} utilized secure multiparty XOR and secure multiparty logical AND. Imran et al.~\cite{imran2023secure} also designed a secure MPSI in the quantum domain. It uses exact the quantum period-finding algorithm (EQPA) as a subroutine. They constructed a quantum multiparty private set intersection (PSI) by transforming the PSI problem into the problem of computing the GCD.
The protocol given in~\cite{imran2023secure} uses rather complicated quantum operators. Since it uses Shor's algorithm and other quantum heavy resources, it is not practical to realize on a large scale with existing quantum computing hardware capabilities. The communication complexity of~\cite{imran2023secure} is higher than the communication complexity of {\sf qMPSI}. The design of~\cite{shi2023edge} is comparable to the design of {\sf qMPSI} in terms of quantum resources used and quantum operators employed. Also, {\sf qMPSI} has a slight edge over~\cite{shi2023edge} in communication complexity. The communication complexity of~\cite{shi2023edge} increases at a quadratic rate in terms of the number of parties involved. {\sf qMPSI} is based on {\sf qOLE} which is seen as a more fundamental building block for secure multiparty computation. {\sf qMPSI} is generic in the sense that any quantum secure OLE can be used to instantiate the protocol. The results of the comparison are summarized in Table~\ref{tab:mpsi-compare}.
\begin{table*}
	\centering
	\begin{tabular}{|p{2.8cm}|p{2.8cm}|p{2.8cm}|p{2.8cm}|p{2.8cm}|}\hline
		Protocols  & Quantum resources & Quantum operators & Quantum measurements & Quantum Communication cost \\ \hline
		\cite{imran2023secure} &  Shor's algorithm quantum resources & Quantum Fourier transforms and CNOT gate & Projective measurements in higher dimensional & $\mathcal{O}(m^2(\log_2{p})^2)$\\ \hline
		\cite{shi2023edge} & Single qubits photons & single qubit $P$ gate and $H$ gate & Projective measurements in $2$ -dimension Hilbert space & $\mathcal{O}(m^2 n)$\\
		\hline
		{\sf qMPSI} & Single qubit photons & single qubit $X$ gate and $Z$ gate & Projective measurements in $2$-dimension Hilbert space & $\mathcal{O}(mn\log_2{p})$ \\ \hline
	\end{tabular}
	\caption{Comparative summary of {\sf qMPSI}}
	\label{tab:mpsi-compare}
\end{table*}

\section{Conclusion}
In this paper, quantum secure protocols for secure multiparty computation (MPC) were designed. The previously introduced number theoretic based MPC protocols are not secure because quantum algorithms like Shor's algorithm~\cite{shor1999polynomial} can be employed to break their security. Firstly, a design of the information theoretic secure quantum secure oblivious linear evaluation ({\sf qOLE}) is proposed. Next, {\sf qOLE} is used to develop a quantum secure multiparty private set intersection (namely {\sf qMPSI}). We believe that it is worth investigating and  designing other types of MPC protocols in the quantum domain.

\section*{Declarations}

\subsubsection*{Funding Details} This work was supported by the ``International Mathematical Union (IMU) and the Graduate Assistantships in Developing Countries (GRAID) Program".
\subsubsection*{Conflict of interest} The authors state that they have not known competing financial interests or personal connections that may seem to have influenced the work described in this study.

\subsubsection*{Data availability} Data sharing is not applicable to this article as no new
data were generated or analyzed to support this research.

\bibliographystyle{IEEEtran}
\bibliography{ref}

\vspace*{-1cm}
\begin{IEEEbiography}[{\includegraphics[width=1in,height=1.15in,clip,keepaspectratio]{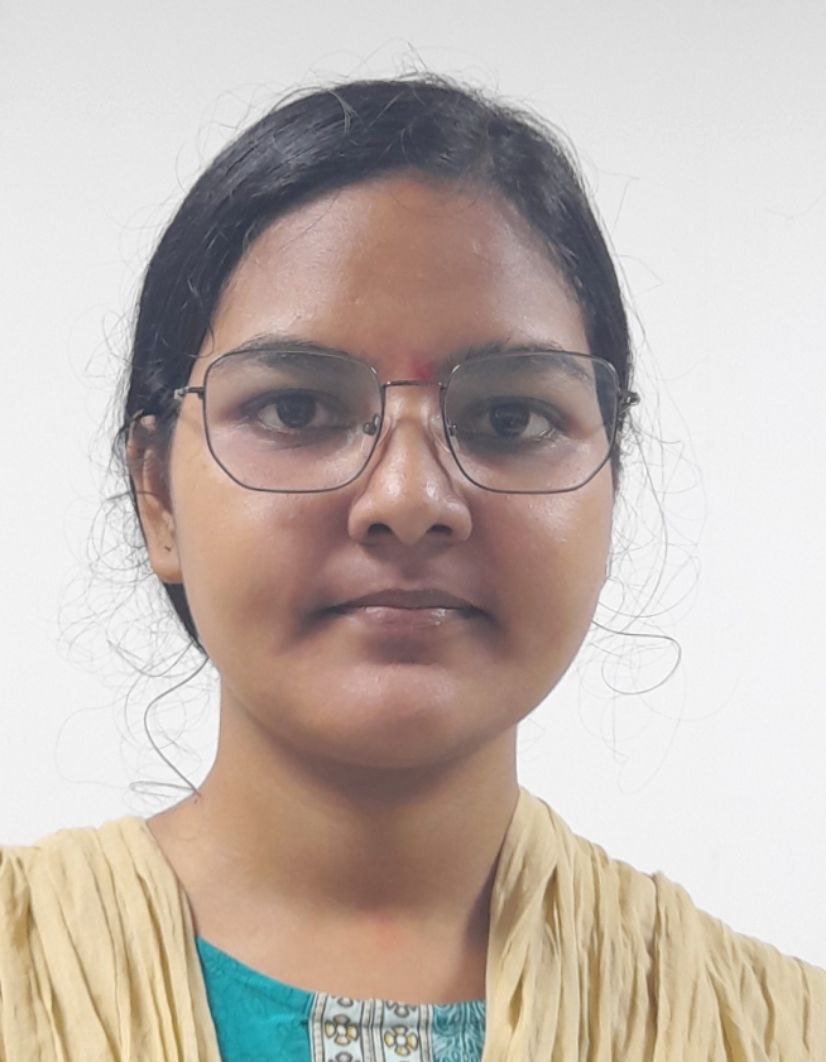}}]{Tapaswini Mohanty} is currently working as a Ph.D. student in the Department of Mathematics,	NIT Jamshedpur. He completed his Masters degree from Institute of Science, BHU in 2019. Her	research interests include Quantum Cryptography, and Private Set Operations.
\end{IEEEbiography}

\vspace*{-1cm}
\begin{IEEEbiography}[{\includegraphics[width=1in,height=1.15in,clip,keepaspectratio]{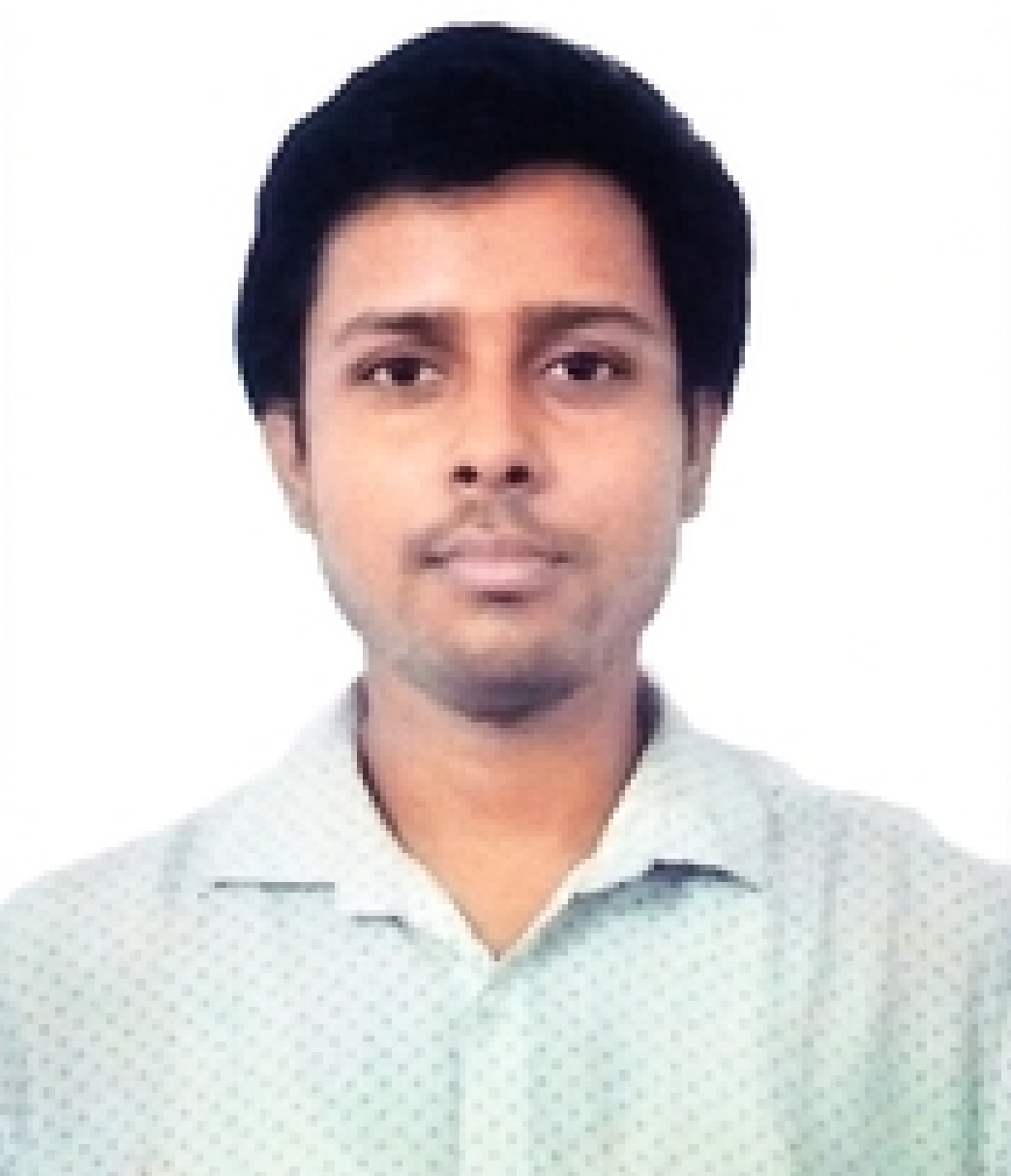}}]{Vikas Srivastava} is working as research scholar in the Department of Mathematics,  National Institute of Technology, Jamshedpur, India. He has completed his B.S.-M.S. dual degree in Mathematics from Indian Institute of Science Education and Research (IISER), Mohali, India in 2017. His research interests include cryptography, network security and blockchain technology.
\end{IEEEbiography}

\vspace*{-1cm}
\begin{IEEEbiography}[{\includegraphics[width=1in,height=1.15in,clip,keepaspectratio]{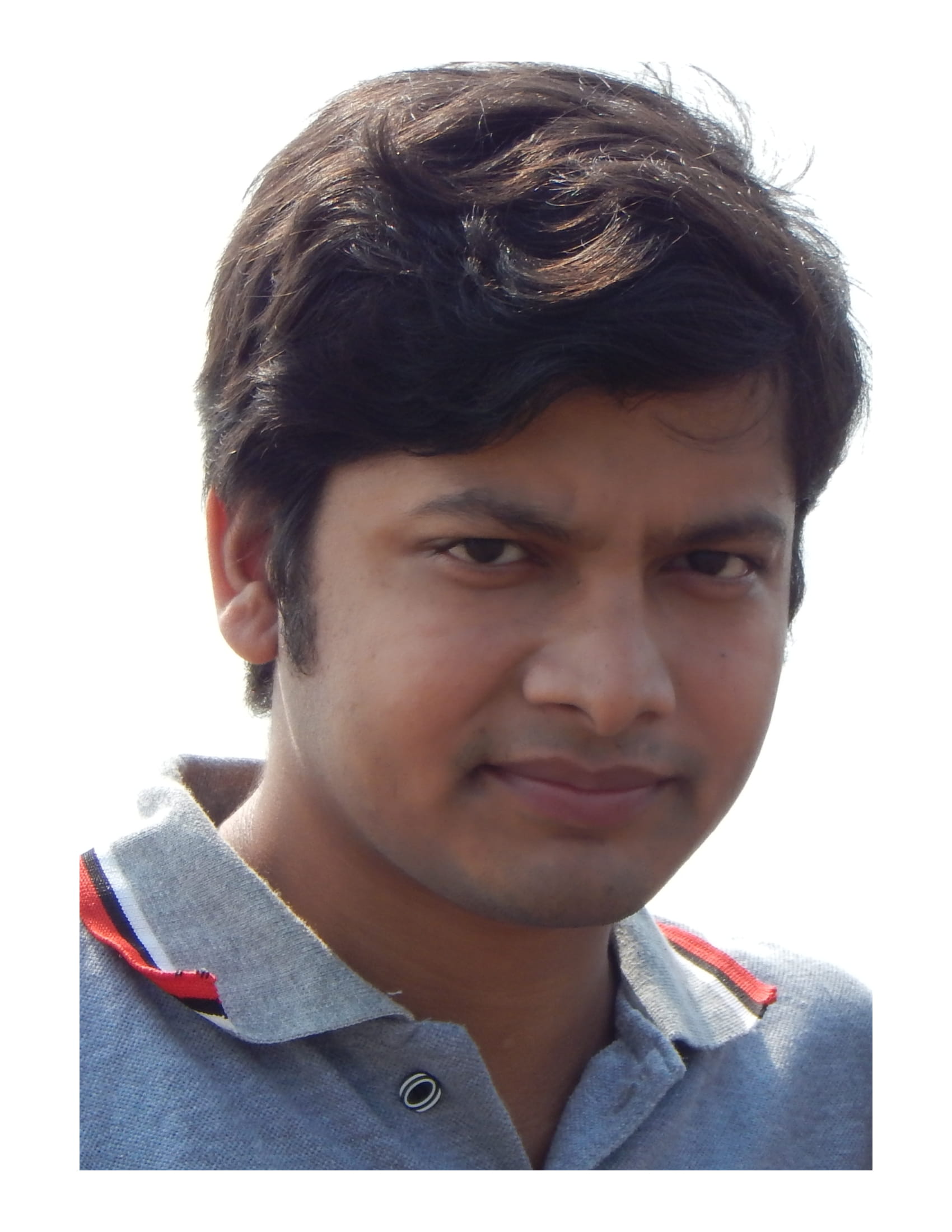}}]{Sumit Kumar Debnath} received a M.Sc. degree in Mathematics from IIT
	Kharagpur in 2012, and also a Ph.D. degree in Cryptology and Network Security from the Department of Mathematics, IIT Kharagpur in 2017. He is currently an Assistant Professor at the Department of Mathematics, National Institute of Technology, Jamshedpur, India.  He is a life member of the Cryptology Research Society of India (CRSI). His research interests include multivariate cryptography, lattice-based cryptography, network security and blockchain.  He has published more than 28 papers in international journals and conferences in his research areas.
\end{IEEEbiography}

\vspace*{-1cm}
\begin{IEEEbiography}[{\includegraphics[width=1in,height=1.15in,clip,keepaspectratio]{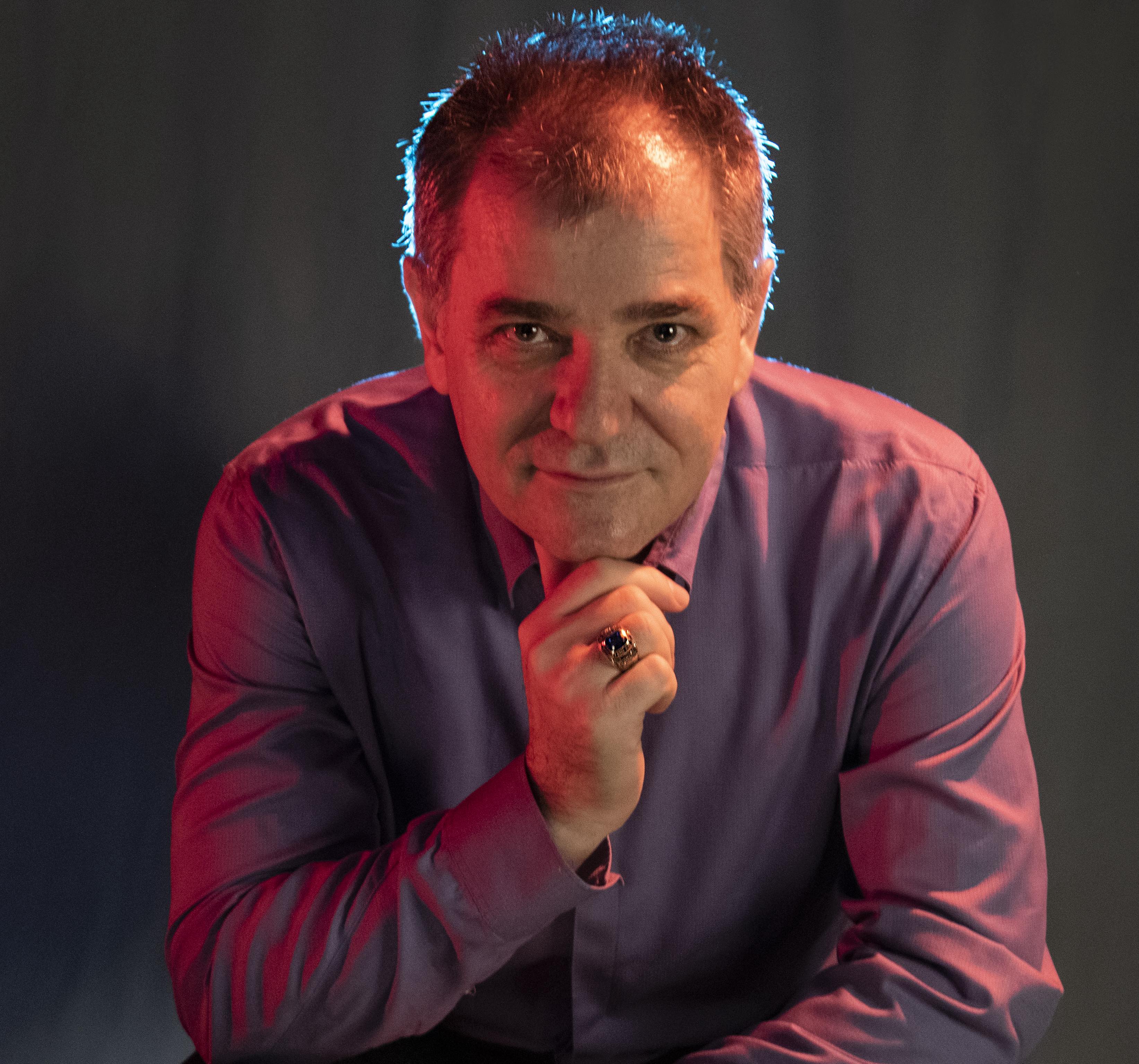}}]{Pantelimon St\u anic\u a} received a Masters of Arts in Mathematics from Bucharest University, Romania in 1992. He received a Ph.D./Doctorate in Algebra from the Institute of Mathematics of the Romanian Academy in 1998 also a Ph.D. degree in Mathematics from State University of New York at Buffalo in 1998. He is currently an Professor and Manager of the Secure Communication program in the Department of Applied Mathematics at Naval Postgraduate School, Monterey, CA 93943, USA. He has published more than 150 papers in refereed journals. He has also published more than 35 papers in refereed conference proceedings. He has also won the 2021 George Boole International Prize for considerable contributions to the theory of Boolean functions.
\end{IEEEbiography}

\end{document}